\documentclass[twoside]{article}
\usepackage{amsfonts,amssymb,amsbsy,textcomp,marvosym,picins,amsmath,caption,threeparttable,amsthm,subfigure,float,lastpage,lscape}
\usepackage{eurosym,mathrsfs,fancyhdr,CJK,multicol,graphics,indentfirst,color,bm,upgreek,booktabs,graphicx,multirow,warpcol}
\usepackage{epstopdf}

\newtheorem{lemma}{Lemma}

\usepackage[linesnumbered,ruled,vlined]{algorithm2e}
\usepackage{algpseudocode}
\usepackage{amsmath}

\usepackage{multirow}

\usepackage{longtable}

\usepackage{graphicx}

\def\footnoterule{\kern 1mm \hrule width 10cm \kern 2mm}

\allowdisplaybreaks
\catcode`@=11
\def\title#1{\vspace{3mm}\begin{flushleft}\vglue-.1cm\Large\bf\boldmath\protect\baselineskip=18pt plus.2pt minus.1pt #1
\end{flushleft}\vspace{1mm} }
\def\author#1{\begin{flushleft}\normalsize #1\end{flushleft}\vspace*{-4pt} \vspace{3mm}}

\def\jz#1#2{{$^{\footnotesize\textcircled{\tiny #1}}$\let\thefootnote\relax\footnotetext{\!\!$^{\footnotesize\textcircled{\tiny #1}}$#2}}}
\catcode`@=11
\def\section{\@startsection{section}{1}{\z@}%
 {-3ex \@plus -.3ex \@minus -.2ex}%
 {2.2ex \@plus.2ex}%
{\normalfont\normalsize\protect\baselineskip=14.5pt plus.2pt minus.2pt\bfseries}}
\def\subsection{\@startsection{subsection}{2}{\z@}%
 {-3ex\@plus -.2ex \@minus -.2ex}%
 {2ex \@plus.2ex}%
{\normalfont\normalsize\protect\baselineskip=12.5pt plus.2pt minus.2pt\bfseries}}
\def\subsubsection{\@startsection{subsubsection}{3}{\z@}%
 {-2.2ex\@plus -.21ex \@minus -.2ex}%
 {1.4ex \@plus.2ex}
{\normalfont\normalsize\protect\baselineskip=12pt plus.2pt minus.2pt\sl}}

\begin{document}
\begin{CJK*}{GBK}{song}
\thispagestyle{empty}
\vspace*{-13mm}
\noindent {\small Journal of computer science and technology: Instruction for authors.
JOURNAL OF COMPUTER SCIENCE AND TECHNOLOGY}
\vspace*{2mm}

\title{Novel Algorithms for Efficient Mining of Connected Induced Subgraphs of a Given Cardinality}

\let\thefootnote\relax\footnotetext{{}\\[-4mm]\indent\ Regular Paper}
\author{Shanshan Wang, Chenglong Xiao, Department of Computer Science, Shantou University}

\noindent {\small\bf Abstract} \quad  {\small \textcolor{blue}{Mining subgraphs with interesting structural properties from networks (or graphs) is a computationally challenging task. In this paper, we propose two algorithms for enumerating all connected induced subgraphs of a given cardinality from networks (or connected undirected graphs in networks). The first algorithm is a variant of a previous well-known algorithm. The algorithm enumerates all connected induced subgraphs of cardinality $k$ in a bottom-up manner. The data structures that lead to unit time element checking and linear space are presented. Different from previous algorithms that either work in a bottom-up manner or a reverse search manner, an algorithm that enumerates all connected induced subgraphs of cardinality $k$ in a top-down manner is proposed. The correctness and complexity of the top-down algorithm are theoretically analyzed and proven. In the experiments, we evaluate the efficiency of the algorithms using a set of real-world networks from various fields. Experimental results show that the variant bottom-up algorithm outperforms the state-of-the-art algorithms for enumerating connected induced subgraphs of small cardinality, and the top-down algorithm can achieve an order of magnitude speedup over the state-of-the-art algorithms for enumerating connected induced subgraphs of large cardinality.}}

\vspace*{3mm}


\vspace*{4mm}

\end{CJK*}
\baselineskip=18pt plus.2pt minus.2pt
\parskip=0pt plus.2pt minus0.2pt
\begin{multicols}{2}

\section{Introduction}
Mining the subgraphs with interesting properties from networks has attracted much attention in recent years. The subgraph enumeration problem arise in many applications, especially in the area of bioinformatics. For example, to detect network motifs, an algorithm for enumerating all connected induced subgraphs with a given cardinality from biological networks was proposed in \cite{1}. Maxwell et al. presented a non-heuristic algorithm for identifying connectivity in molecular interaction networks by enumerating connected induced subgraphs of molecular interaction networks\cite{2}. The authors of \cite{3} proposed an efficient algorithm for enumerating all connected induced subgraphs as the backbone of mining all maximal cohesive subgraphs representing biological subnetworks. The subgraph enumeration algorithms were also introduced and applied in several other fields including computer science, electronics, communication, chemistry, to name a few. Examples of such applications include retrieving connected subgraphs for keyword search over RDF graph in information retrieval\cite{4}, identifying connected convex subgraphs as custom instructions in the design of extensible processors\cite{41}, generating virtual software-defined networks requests in 6G core networks \cite{42}, performing consistency analysis in constrained processing\cite{5}, determining the maximally coarse partitions to exploit categorical structure\cite{51} and predicting crystal structures in models of chiral molecules \cite{52}.

Enumerating all connected induced subgraphs of a given cardinality from a graph (or a network) is essentially a computationally difficult problem. The number of all connected induced subgraphs of a fixed cardinality $k$ is exponentially large in $k$. The previous study shows that there exists an upper bound of $n\cdot\frac{(e\Delta)^{k}}{(\Delta-1)k}$ on the number of connected induced subgraphs of cardinality $k$, where $\Delta$ is the maximum degree of a connected graph $G$ and $n$ is the number of vertices in $G$\cite{6}. In this paper, to be consistent with the notations presented in \cite{7}, we will denote the task of enumerating all connected induced subgraphs of cardinality $k$ of $G$ as E-CISE.

In this paper, we first present a variant algorithm of the previous well-known algorithm\cite{1}. In the algorithm, we use a guarding set to record the visited vertices. The algorithm incrementally enumerates the required subgraphs by adding a neighbor vertex of the current subgraph (the neighbor vertex must not be in the guarding set). In such a way, each connected induced subgraph is enumerated and no duplicates are produced. Different from the existing algorithms in the literature that either work in a bottom-up manner or in a reverse search way, we propose a top-down algorithm that enumerates all connected induced subgraphs of size $k$ by deleting so-called deletable vertices (vertices are non-articulation points) from large connected subgraphs until the size of the current subgraph is $k$.

The rest of the paper is organized as follows. The state-of-the-art algorithms are introduced in Section II. Section III first presents the notations and preliminaries. Then, we describe the proposed variant algorithm and the proposed top-down algorithm in detail. The proof of the correctness and the complexity of the top-down algorithm is also presented. The experimental comparison and analysis are provided in Section IV. Section V discusses future work. Finally, Section VI concludes the paper.

\section{Related Works}

Since the number of connected induced subgraphs can be exponential in the worst case, efficient enumeration of such subgraphs is quite necessary. Much work has been spent on enumerating connected induced subgraphs in recent years. Scanning the literature for algorithms enumerating connected induced subgraphs, most known algorithms perform the enumeration procedure in a bottom-up manner: starting with a single-vertex subgraph $S=\{v_i\}$, larger connected induced subgraphs are successively formed by adding neighbor vertices. However, the crucial problem in the enumeration procedure is that subgraphs may be enumerated multiple times. Hence, how to avoid duplicates is the main domination of the efficiency of these algorithms.

Wernicke introduced an algorithm for enumerating all size-k connected induced subgraphs \cite{1}. The algorithm starts by assigning each vertex a number as a label. The larger subgraphs are generated by absorbing vertices from $V_{Extension}$ to smaller subgraphs. The vertices in $V_{Extension}$ set should satisfy two conditions: their assigned number must be larger than that of $v$ ($v$ is the vertex that the larger subgraphs originated from) and they can only be the neighbors of the newly added vertex but not to a vertex in parent subgraph. The most recent literature\cite{7} analyzed Wernicke's algorithm (denoted as $Simple$) and showed that the algorithm solves the enumeration of all size-k connected induced subgraphs with a delay of $O(k^{2}\Delta)$, where $\Delta$ is the maximum degree of $G$. A variant of $Simple$ called $Simple-Forward$ was also proposed in \cite{7}. The variant algorithm selects the vertex from the front of $V_{Extension}$ instead of from the back as in $Simple$. A slightly different algorithm called $Pivot$ was proposed in \cite{8}. In the enumeration procedure, the algorithm expands the subgraph by adding the neighbors of the selected active vertex.

\begin{figure*}[t]
  \centering
  \includegraphics[width=0.30\hsize=0.8]{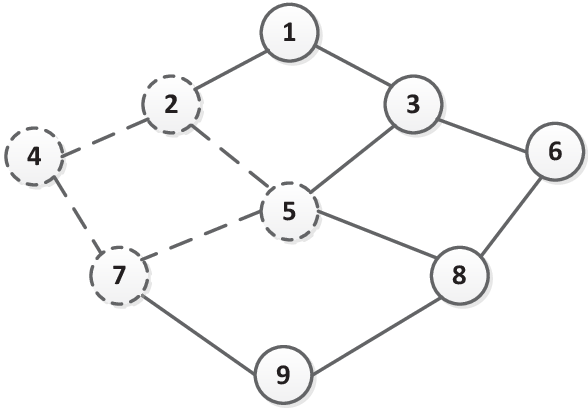}
  \caption{An example of undirected graph with 9 vertices and 12 edges.}
  \label{f2}
\end{figure*}

Uno presented a binary search tree-based algorithm for enumerating all connected induced subgraphs\cite{9}. For each given vertex $v$, the connected induced subgraphs are divided into two groups: subgraphs containing $v$ and subgraphs not containing $v$. The former group can be further recursively enumerated by dividing the subgraphs into two groups based on another neighbor vertex $u$. The latter group can be recursively enumerated by deleting $v$. The author has shown that the time complexity is $O(1)$ for each solution by amortization. However, the solutions can only be produced at the leaves of the search tree.

In \cite{3} the authors refer to the vertex $s$ with the smallest vertex label in a vertex set as the anchor vertex and the vertex $u$ with the longest shortest path to the anchor vertex as the utmost vertex. The algorithm expands the subgraphs by adding only the neighbor vertex $v$ that either has a longer distance to $s$ than the distance from $s$ to $u$ or is lexicographically greater than $u$ if $v$ and $u$ have the same distance to $s$.

However, despite previous work to propose new algorithms to enumerate connected induced subgraphs of cardinality $k$, extensive experiments carried out in recent paper \cite{7} showed that the variant algorithm proposed in \cite{7} $Simple-Forward$ is slightly faster than the algorithm $Simple$ \cite{1} in practice, and is the fastest one among all the existing algorithms.

The aforementioned algorithms enumerate the connected induced subgraphs by incrementally expanding the subgraphs. These algorithms can be classified as bottom-up algorithms. Elbassuoni introduced an algorithm specifically targeting to enumerate connected induced subgraphs of cardinality $k$ based on the supergraph method using reverse search\cite{11}. The algorithm initially generates a connected induced subgraph with $k$ vertices. Then, the algorithm (denoted as $RwD$) tries to enumerate every neighbor node in the supergraph by replacing a vertex in the current subgraph with a neighbor vertex of the current graph. The $RwD$ algorithm has a delay of $\mathcal{O}(k\cdot min(n-k,k\Delta)(k(\Delta+\log{n}))$ and requires $O(n+m+k|S|)$ space for storing, where $|S|$ is the number of nodes in a super graph $S$, $n$ is the number of vertices and $m$ is the number of edges. To reduce the required space, a variant of $RwD$ (denoted as $RwP$) was also introduced. The $RwP$ algorithm has a larger delay of $\mathcal{O}((k\cdot min(n-k,k\Delta))^{2}(k(\Delta+\log{n}))$ but it requires $O(n+m)$ space. However, the extensive experiments carried out in \cite{7} reveal that both $RwD$ and $RwP$ are significantly slower than the aforementioned algorithms that work in a bottom-up manner.

\textit{Problem E-CISE}: Given an undirected graph $G = (V,E)$, enumerate all connected vertex set of cardinality $k$ , $CIS(G,k)=\{U|U\subseteq V~and~G(U)~is~connected~and~|U|=k\}$.

\begin{algorithm*}[t]
  \caption{The $VSimple$ Algorithm}
  \KwIn{An undirected graph $G=(V,E)$, $Y$: The guarding set recording the considered vertices}
  \KwOut{A set of enumerated connected induced subgraphs of cardinality $k$}
  \tcc{enumerate all connected induced subgraphs of cardinality $k$ in a bottom-up manner}
  \textbf{Procedure} $Enumerate(G(V,E))$\\
  \For{each vertex $v_i\in V$}
  {
     $VSimple(\{v_i\},N(v_i)\backslash Y,Y)$\;
     $Y = Y \cup \{u\}$\;
  }

  \tcc{recursively  expand the subgraph by adding a neighbor not in $Y$}
  \textbf{Procedure} $VSimple(C,N,Y)$\\

   \If{$|C|=k$}
   {
      output $C$\;
      return True\;
   }
  hasIntLeaf = false\;
  \For{each vertex $u\in N$}
  {
    $C' = C \cup \{u\}$\;
    delete $u$ from $N$\;
    $N' = N \cup (N(u)\backslash C'\backslash Y)$\;
    \If{ $VSimple(C',N',Y)$ = True}
    {
        hasIntLeaf:=True\;
    }
    \Else
    {
       break\;
    }
    $Y = Y \cup \{u\}$\;
    \If{$|V|-|Y|<k$}
   {
       break\;
   }
  }
  return hasIntLeaf\;
\end{algorithm*}

\section{The Proposed Algorithms}

We start this section by presenting some necessary notations and preliminaries used in the proposed algorithms. Subsection B briefly presents a variant algorithm of the previous well-known algorithm \cite{1}. Subsection B presents the proposed top-down algorithm in detail. The correctness and the time complexity of the algorithm are also presented in Subsection C.

\subsection{Preliminaries}
Let $G=(V,E)$ be an undirected graph, where $V=\{v_1,v_2...,v_n\}$ is the set of vertices, and $E\subseteq V\times V$ is the set of edges. The cardinality of a (sub)graph is the number of its vertices. For a vertex set $U\subseteq V$, $G(U)$ denotes the subgraph of $G$ induced by $U$. All the edges of $G$ with endpoints in $U$ are included in $G(U)$. A graph is said to be connected if and only if any two vertices in the graph are connected to each other by at least one path. The set of neighbors of a vertex $v$ is denoted by $N(v):=\{u|\{u,v\}\in E\}$. The set of neighbors of a vertex set $X$ is denoted by $N(X):=\cup_{v\in X}N(v)\setminus X$.

Figure 1 shows an example of an undirected graph. The dashed edges and vertices represent the connected subgraph induced by $\{2,4,5,7\}$. The set of neighbors of $2$ is $N(2)=\{1,4,5\}$. The set of neighbors of $\{2,4,5,7\}$ is $N(\{2,4,5,7\})=\{1,3,8,9\}$.

\subsection{$VSimple$ Algorithm}
In this subsection, we propose a variant algorithm for enumerating all connected induced subgraphs of cardinality $k$ from an undirected graph. In this paper, we assume the input undirected graph is connected, if not we can simply deal with each connected component separately. The pseudo-code of the algorithm is presented in Algorithm 1. The algorithm accepts the undirected graph $G(V, E)$ as input, and outputs all the connected induced subgraphs of cardinality $k$ (lines 6-7, Algorithm 1).

\begin{figure*}[t]
  \centering
  \includegraphics[width=0.48\hsize=0.8]{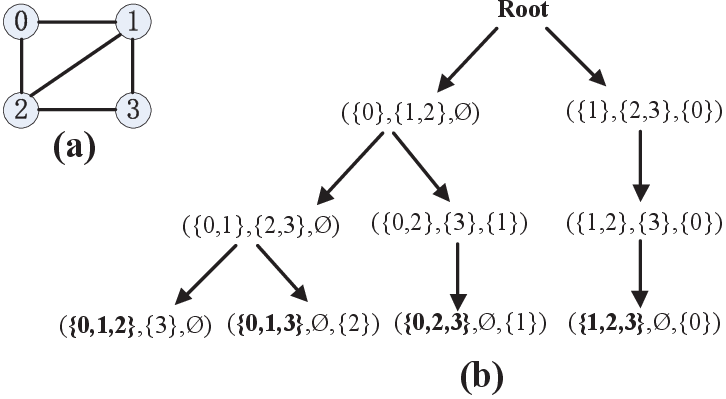}
  \caption{Enumeration tree of $VSimple$ for enumerating all connected induced subgraphs of cardinality 3.(a) a sample graph (b) the enumeration tree on the sample graph. Nodes of the tree indicate enumeration nodes, i.e. the call $VSimple(C,N,Y)$. The enumerated connected induced subgraphs of cardinality 3 are highlighted in bold.}
  \label{f2}
\end{figure*}

The algorithm works in a bottom-up manner (we denote the algorithm as $VSimple$). All the enumerated connected induced subgraphs are recursively generated. We initially generate every one-vertex subgraph $\{v_i\}$ (lines 2-3, Algorithm 1). Then, we recursively call $VSimple$ to enlarge the subgraph $X$ by adding a neighbor node of $C$. As we know, identifying and ruling out duplicates from a large set is a time-consuming job. To avoid generating duplicates, a guarding set $Y$ is used to record the already considered vertices (in other words, the vertices in $Y$ can not be included in the following iterations and their subsequent recursive calls). In the procedure of expanding $C$, only the neighbor vertex that is not in $Y$ can be added to $C$ to form a larger subgraph $C'$. The set of neighbor vertices of $C'$ is updated by adding the vertices neighbored to the newly added vertex $u$ but not in $C'$ and $Y$ to $N$ (line 13, Algorithm 1). The algorithm stops when the size of the current subgraph is $k$. In this algorithm, we also perform the pruning rule proposed in \cite{7} using a Boolean flag hasIntLeaf (lines 9,15,21, Algorithm 1), and the k-component rule to speed up the enumeration (lines 19-20, Algorithm 1).

We consider a simple example to illustrate the enumeration tree of the proposed algorithm in Figure 2. In this example, we enumerate all connected induced subgraphs of cardinality $3$. The algorithm starts by searching all connected induced subgraphs including vertex $0$. At the search node $VSimple(\{0\},\{1,2\},\emptyset)$, we will first consider all connected induced subgraphs including $\{0\}$ and vertex $1$. After that, all connected induced subgraphs including $\{0\}$ and vertex $2$ are enumerated by deleting vertex $1$ (adding vertex $1$ to $Y$). After enumerating all connected induced subgraphs containing vertex $0$, vertex $0$ is added to $Y$. We then perform the search of all connected induced subgraphs including vertex $1$ and excluding vertex $0$.

\begin{figure*}[t]
  \centering
  \includegraphics[width=0.78\hsize=0.8]{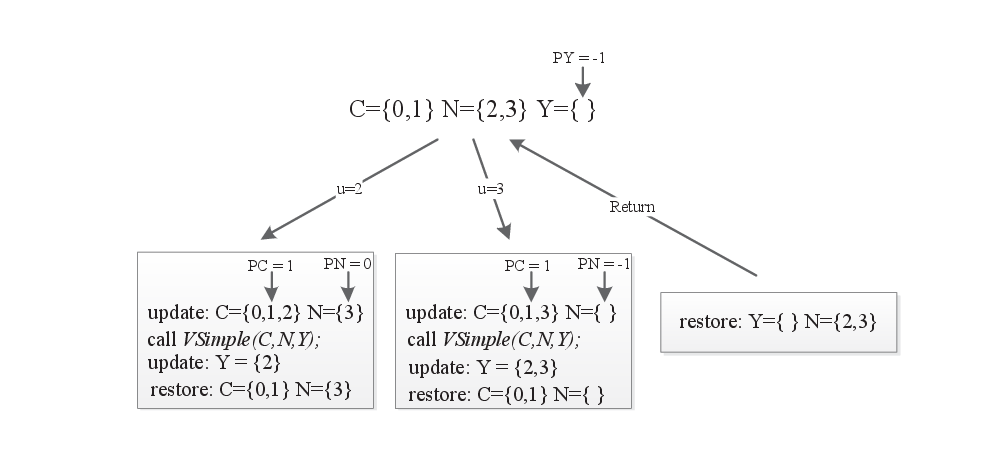}
  \caption{An example of using the proposed data structures for the enumeration node $\{0,1\},\{2,3\},\emptyset$ and its children nodes presented in Fig.2.}
  \label{f2}
   \vspace{-2mm}
\end{figure*}

As the algorithm is a variant of the previous algorithm\cite{1,8}, we omit the proof of the correctness and the complexity of $VSimple$. We highlight the differences between $VSimple$ and $Simple$. In $Simple$ algorithm, the set $N'$ is updated by setting $N'=N\cup (N(u)\setminus N[C])$ where $N[C]$ denotes the closed neighborhood of $C$ ($N[C]:=\cup_{v\in C}N(v)$), whereas the $VSimple$ algorithm updates the set $N'$ by setting $N'=N\cup (N(u)\setminus C'\setminus Y)$ where $Y$ is a set of forbidden vertices.

Now, we describe the data structures used in the implementation to achieve a linear space of $O(n+m)$, and a delay of $O(k^{2}\Delta)$ as achieved for $Simple$\cite{7}. To achieve the linear space of $O(n+m)$, we can use arrays with dynamic length or linked lists to store the original graph $G$ (adjacent lists of $G$), sets $C$, $N$, and $Y$. These arrays or lists are declared as global variables in the implementation. The maximal length of these arrays or lists is $|V|$. Inside each call of $VSimple$, we have a for loop. Inside each iteration of the for loop, we update $C$, $N$, and $Y$ by adding new vertex (or vertices) to the end of the used arrays or lists. After the recursive call of the current iteration (line 14, Algorithm 1), we should restore $C$ and $N$ for the following iterations. To restore $C$ and $N$, we can set pointers that indicate the end position of the arrays or lists before updating (before line 11, Algorithm 1). With the pointers, we can restore $C$ and $N$ by removing the newly added elements that lie between the end position (inclusive) to the pointer (exclusive). For $Y$, we restore it after all the iterations are performed. To restore $Y$, we also set a pointer that indicates the end position of $Y$ before the execution of the for loop (before line 10, Algorithm 1). Aiming at the delay of $O(k^{2}\Delta)$, it is essential to ensure that the update of $N$ can be performed in $O(\Delta)$. To update $N$, we check every element $v$ in $N(u)$. If $v$ is not in $C'$ and $Y$, we add it to $N$. Thus, unit time checking whether an element $v$ is in $C'$ and $Y$ is required. In the implementation, we use auxiliary Boolean arrays of $C$ and $Y$, where the index of the Boolean arrays corresponds to a vertex, and the true value represents the vertex in the set. The auxiliary Boolean arrays are updated and restored simultaneously with the corresponding arrays or linked lists.

Figure 3 presents an example of using the proposed data structures for the enumeration node $(\{0,1\},\{2,3\},\{\})$ and its children nodes presented in Fig.2. As $N=\{2,3\}$, the enumeration node $(\{0,1\},\{2,3\},\{\})$ has two iterations ($u=2$ and $u=3$). Before entering the first iteration ($u=2$), we set a pointer for $Y$ to indicate the end position of $Y$. As $Y$ is empty, we have $PY=-1$. Inside the iteration of $u=2$, we first get the end position of $C$ and $N$ by setting a pointer for $C$ and a pointer for $N$ respectively ($PC=1$ and $PN=0$), update $C$ by adding vertex $2$ to the end of $C$ and update $N$ by deleting $2$ from $N$. Then, we call $VSimple(C, N, Y)$. After the call of $VSimple(C, N, Y)$, we update $Y$ by adding vertex $2$ to the end of $Y$. Before entering the next iteration, we restore $C$ and $N$ using the pointers $PC$ and $PN$ (the vertices that lie between the current end position (inclusive) and the pointer (exclusive) are removed). For the iteration $u=3$, we also perform the update and the restoration of sets $C$,$N$, and $Y$ in the same way. After all iterations are performed and before returning to the enumeration node $(\{0,1\},\{2,3\},\{\})$, we restore $Y$ by removing the vertices that lie between the current end of $Y$ and the pointer $PY$, and restore $N$ by adding the vertices that just removed from $Y$. To support unit time checking whether an element $v$ is in sets $C$ or $Y$, we use a global Boolean array for each set. As an example, for $Y=\{1\}$, its corresponding Boolean array is $BY = \{false,true,false,false\}$. The length of these Boolean arrays is the size of $G$. With these auxiliary Boolean arrays, every vertex $v \in N(u)$ can be checked in unit time.

\begin{algorithm*}[t]
  \caption{The $TopDown$ Algorithm}
  \KwIn{An undirected graph $G=(V,E)$, $Y$: The guarding set recording the deleted vertices}
  \KwOut{A set of enumerated connected induced subgraphs of cardinality $k$}
  \tcc{enumerate all connected induced subgraphs of cardinality $k$ in a top-down manner}
  \textbf{Procedure} $Enumerate(G(V,E))$\\
  $N = findNonArticulationPoints(G)$\;
  $TopDown(G,N,Y)$;

  \tcc{recursively reduce the subgraph by deleting a vertex not in $Y$}
  \textbf{Procedure} $TopDown(C,N,Y)$\\

   \If{$|C|=k$}
   {
      output $C$\;
   }
  \For{each vertex $u\in N$}
  {
    $C' = C \backslash \{u\}$\;
    $X = findNonArticulationPoints(C')$\;
    $N' = X \backslash Y$\;
    $TopDown(C',N',Y)$\;
    $Y = Y \cup \{u\}$\;
    \tcc{look-ahead rule}
    \If{$|Y|=k$}
   {
       \If{isConnected(Y)}{
          output $Y$\; 
       } 
       break\;
   }
  }
\end{algorithm*}

\subsection{$TopDown$ Algorithm}
In this subsection, we present a novel algorithm that enumerates all connected induced subgraphs in a top-down manner. The algorithm accepts the original graph $G$ as input and recursively deletes vertices until the size of subgraphs is $k$. To ensure the connectivity of subgraphs, we only consider the vertices that are non-articulation points (lines 2,9, Algorithm 2). A vertex is said to be an articulation point in a connected graph if the removal of the vertex and its associated edges disconnects the graph. The articulation points (or non-articulation points) of subgraph $C$ can be found in $\mathcal{O}(l\cdot\Delta)$ time by calling the algorithm proposed by Tarjan \cite{12}, where $l$ is the size of $C$ and $\Delta$ is the maximum degree of $G$. In this paper, we call non-articulation points deletable vertices. Inside each call, we iteratively delete the vertices that are deletable vertices of the current subgraph and are not in $Y$ (line 8,11, Algorithm 2). Once a vertex $u$ is deleted from $C$ (line 8, Algorithm 2), we add $u$ to the guarding set $Y$ for the rest of the iterations and their subsequent enumerations (line 12, Algorithm 2). In this sense, the guarding set $Y$ can be viewed as the set of vertices that must be included in the following enumerations. Owing to the guarding set $Y$, we can ensure that each connected induced subgraph is enumerated at most once. Since $Y$ represents the set of vertices that must be contained, we can prune the search space by breaking the iterations when the size of $Y$ is larger than $k$. To further speedup the search by avoiding some recursive calls, we judge the connectivity of $Y$ in advance (the connectivity of $Y$ can be determined in $\mathcal{O}(k\cdot\Delta$) when $|Y|=k$. If $Y$ is connected, we output it and then break the iterations (lines 13-16, Algorithm 2). The pruning rule is called as a look-ahead rule.

\begin{figure*}[t]
  \centering
  \includegraphics[width=0.62\hsize=0.8]{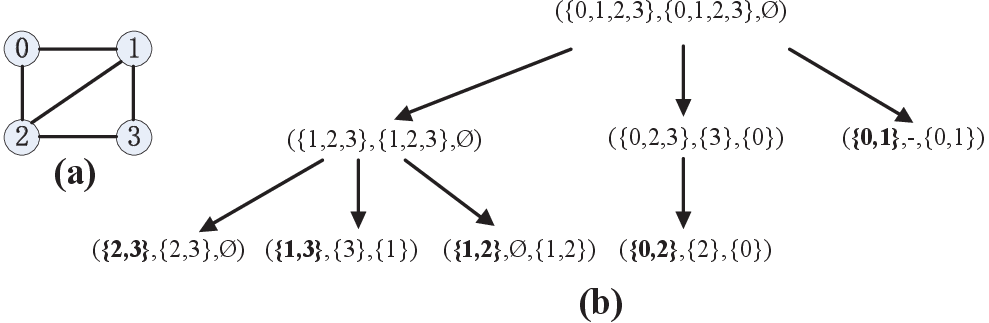}
  \caption{Enumeration tree of $TopDown$ for enumerating all connected induced subgraphs of cardinality 2.(a) a sample graph (b) the enumeration tree on the sample graph. Nodes of the tree indicate enumeration nodes, i.e. the call $TopDown(C,N,Y)$. The enumerated connected induced subgraphs of cardinality 2 are highlighted in bold.}
  \label{f2}
\end{figure*}

Figure 4 illustrates an example of the enumeration tree of the $TopDown$ algorithm. First of all, the algorithm finds the vertices that are non-articulation points of $G$ ($\{0,1,2,3\}$). Then, the algorithm tries to find all the connected induced subgraphs of cardinality 2 excluding vertex $0$ by deleting vertex $0$. After deleting $0$, we have a connected subgraph $\{1,2,3\}$. To get all the connected induced subgraphs of cardinality 2, the algorithm further deletes the vertices that are deletable (non-articulation points). At the end of this branch, we will get three connected induced subgraphs of cardinality 2 (\{2,3\},\{1,3\}, and \{1,2\}). After enumerating all the connected induced subgraphs excluding vertex $0$ of cardinality 2, the algorithm carries out a branch to find all the connected induced subgraphs excluding vertex $1$ and including vertex $0$. This branch finds the subgraph $\{0,2\}$. Finally, the algorithm tries to enumerate all the connected induced subgraphs excluding vertex $2$ and including vertices $0,1$. As the number of vertices that must be included already equals the given cardinality $k=2$, we just need to judge the connectivity of the induced subgraph of $\{0,1\}$. The induced subgraph of $\{0,1\}$ is connected, and the algorithm outputs it.

We use similar data structures presented in $VSimple$ for $TopDown$. The original graph,$C$, and $Y$ are stored in global arrays or linked lists. To support unit time checking, we also use auxiliary Boolean arrays. It should be noted that $N$ is updated by only adding new elements to the end of arrays or linked lists in $VSimple$, however, $N$ is updated by removing some elements (after deleting a vertex $u$, some deletable vertices become articulation points) and adding an element (after deleting a vertex $u$, there is at most one articulation point become a deletable vertex). Thus, we use a global array for $N$ and two extra global arrays or linked lists to record these changes.

\begin{lemma}
Let $S(V, E)$ be a connected undirected graph and $u\in V$ is a deletable vertex of $S$, then there is at most one articulation point of $S$ that becomes a deletable vertex of $S\backslash \{u\}$ after deleting $u$ from $S$.
\end{lemma}

\begin{proof}
We prove this by contradiction. Suppose that there are at least two articulation points ($w$ and $v$) of $S$ that are deletable vertices of $S\backslash \{u\}$ after deleting the deletable vertex $u$ from $S$. As $w$ is an articulation point of $S$ and according to the definition of articulation point, we will have at least two connected components that have no edge connecting to each other after removing $w$ and its associated edges. Assume $S_1$ and $S_2$ are the two connected components after removing $w$. It is clear that $u$ can be in only one of the connected components. We further assume that $u$ is in $S_1$. We now claim that $u$ is the sole vertex in $S_1$ ($S_1\backslash \{u\} = \emptyset$). We can prove this claim by contradiction. As $u$ is a deletable vertex of $S$ and $w$ is a deletable vertex of $S\backslash \{u\}$, $S\backslash \{u\}\backslash \{w\}$ is a connected graph. This denotes that the vertex (or vertices) of $S_1\backslash \{u\}$ is connected to the vertex (or vertices) of $S_2$ in $S\backslash \{u\}\backslash \{w\}$. This contradicts the assumption that $S_1$ and $S_2$ are the two connected components that have no edge connecting them. Since $u$ is the sole vertex in $S_1$, this means that $v\notin S_1$, and $u$ has an edge connecting to $w$. This further indicates that $u$ has an edge connected to $w$ but has no edge connected to $v$. For the articulation point $v$ of $S$, similarly, we can also achieve a conclusion: $u$ has an edge connected to $v$ but has no edge connected to $w$. Therefore, we have arrived at a contradiction.

\end{proof}

With Lemma 1, we know that there is at most one articulation point that becomes a deletable vertex after we delete a deletable vertex.  As $TopDown$ deletes one vertex in each enumeration node and the depth of the enumeration tree is $n-k$, the maximal length of the array that records the added vertices is $n-k$. Furthermore, as the vertices removed from $N$ is the subset of the original graph $G$ and the maximal number of vertices added to $N$ is $n-k$, the maximal length of the array that records the removed vertices is $2n-k$.

\begin{table*}[t]
\centering
\caption{The time complexity and space complexity of $Simple$,$VSimple$, $Simple-Forward$ and $TopDown$}

\begin{tabular}{lll}

\hline

Algorithm &    Time Complexity &    Space Complexity \\

\hline
  $Simple$ &  $\mathcal{O}\left((e(\Delta-1))^{(k-1)}(\Delta+k)\cdot n/k\right)$  &	$\mathcal{O}(|V|+|E|)$	 \\

  $VSimple$ &   $\mathcal{O}\left((e(\Delta-1))^{(k-1)}(\Delta+k)\cdot n/k\right)$ 	&$\mathcal{O}(|V|+|E|)$    \\
 
   $Simple-Forward$ &    $\mathcal{O}\left((e(\Delta-1))^{(k-1)}(\Delta+k)\cdot n/k\right)$ 	&$\mathcal{O}(|V|+|E|)$    \\

   $TopDown$ &  $\mathcal{O}\left({n \choose n-k}\cdot n \cdot\Delta\right)$	&$\mathcal{O}(|V|+|E|)$  \\

\hline
\end{tabular}

\end{table*}

\subsubsection{Correctness and Complexity of the $TopDown$ Algorithm}
In this subsection, we prove the correctness and complexity of the proposed algorithm.
\begin{lemma}
Let $G(V, E)$ be a connected undirected graph, then $TopDown$ correctly outputs all connected induced subgraphs of cardinality $k$ from $G$.
\end{lemma}

\begin{proof}
Firstly we show that each subgraph $G(C)$ output by $TopDown$ is a connected induced subgraph. According to the enumeration procedure, each smaller subgraph is generated by deleting a vertex that is a non-articulation point from a larger subgraph. It is clear that each enumerated subgraph is a connected induced subgraph.

Secondly, we prove that if $G(C)$ is a connected induced subgraph of cardinality $k$, then $G(C)$ is output by $TopDown$ at least once. To prove this we will show that for every connected set $C\subset H \subseteq V$, there exists at least a vertex $s\in H\backslash C$ is a non-articulation point of $H$. We prove this claim by contradiction. Suppose that every $s\in H\backslash C$ is an articulation point of $H$. Then, we can split $H$ into at least two connected subgraphs $H_1$ and $H_2$ after deleting $s$. It is clear that $C$ can be only in one of these connected subgraphs. Let's assume $C$ is in $H_1$. We can choose any vertex of $H_2$ as the root to get a spanning tree of $H_2$. We also know that the leaf vertex(s) of the spanning tree can never be an articulation point of $H_2$ and $H$. Thus, a contradiction is arrived. Since the claim is true, we can start from the original graph $G$ by calling $TopDown(C, N, Y)$ with a non articulation point $s \in V\backslash C$. Then we recursively follow the execution path by deleting such non-articulation point $s \in H\backslash C$ until $C$ is output.

Thirdly we prove that if $G(C)$ is a connected induced subgraph of cardinality $k$, then $G(C)$ is output by $TopDown$ at most once. We can prove this by contradiction. We assume $C$ is enumerated twice ($C'~and~C''$). Let further assume $C'$ and $C''$ are enumerated by two enumeration paths $S_1$ and $S_2$ respectively. According to the enumeration procedure, if we track the reverse direction of the two enumeration paths, the two enumeration paths should converge at one point. Starting from that converging point, one of the two enumeration paths generates all the connected induced subgraph(s) excluding a vertex $u$, while the other enumeration path generates all the connected induced subgraph(s) including $u$. Thus, $C'\neq C''$, i.e. a contradiction.

\end{proof}

\begin{lemma}
The running time of $TopDown$ is $\mathcal{O}\left({n \choose n-k}\cdot n \cdot\Delta\right)$, where $\Delta$ is the maximum degree of $G$, and $k$ is the size of all connected induced subgraphs.
\end{lemma}

\begin{proof}
Firstly, we bound the number of nodes in the enumeration tree of $TopDown$. According to the enumeration procedure of $TopDown$, each connected induced subgraph of cardinality $k$ is outputted by deleting $n-k$ vertices. Thus, we may have up to ${n \choose n-k}$ possible solutions. We can also observe that $TopDown$ outputs each connected induced subgraph of cardinality $k$ only at the leaf node of the enumeration tree. Therefore, the maximal number of leaf nodes in the enumeration tree is $\mathcal{O}({n \choose n-k})$. We consider the worst case with $\mathcal{O}({n \choose n-k})$ leaf nodes. Each inner node of the enumeration tree has at least two children ($n-k+1$ children) for $1\leq k\leq n-1$. Hence, the number of inner nodes is smaller than the number of leaf nodes. Consequently, the overall number of enumeration tree nodes of $TopDown$ is $\mathcal{O}\left({n \choose n-k}\right)$.

Secondly, it is sufficient to prove that each call of $TopDown$ without recursive calls is charged with $\mathcal{O}\left(n\cdot\Delta\right)$. It is clear that we can output each connected induced subgraph of cardinality $k$ in $O(k)$ time. To update $X$, we call $findNonArticulationPoints(C')$ to get the vertices that are non-articulation points of $C'$. This call requires at most $\mathcal{O}(n\cdot \Delta)$ time. The update of $C'$, $Y$ and $N'$ in each recursive call require at most $\mathcal{O}(1)$, $\mathcal{O}(1)$ and $\mathcal{O}(n)$ time respectively. We charge these running time to the corresponding child in the enumeration tree. So, we obtain $\mathcal{O}\left(n\cdot \Delta\right)$ time in total per enumeration node. Finally, we conclude that the total running time of $TopDown$ is $\mathcal{O}\left({n \choose n-k}\cdot n \cdot\Delta\right)$.
\end{proof}

\begin{lemma}
   The space complexity of $TopDown$ is $\mathcal{O}(|V|+|E|)$, where $|V|$ is the number of vertices in $G$, and $|E|$ is the number of edges in $G$.
\end{lemma}

\begin{proof}
The space required for storing $G$ is $|V|+|E|$. According to data structures used in the implementation, we know that $C$, $N$, and $Y$  require global sets (linked lists, Boolean arrays, or bit arrays) of maximum size $|V|$. It is obvious that the space complexity of algorithm $TopDown$ is dominated by the space storing the graph. Therefore, the space complexity of algorithm $TopDown$ is $\mathcal{O}(|V|+|E|)$.
\end{proof}

Table 1 presents the time complexity and space complexity of  $Simple$, $VSimple$,  $Simple-Forward$, and  $TopDown$. $Simple$, $VSimple$, and  $Simple-Forward$ have the same time complexity and space complexity. The four algorithms require identical space complexity. By comparing the time complexity, we can observe that the time complexity of $TopDown$ is tighter than the time complexity of $Simple$, $VSimple$, and $Simple-Forward$ in the cases of large $k$.

\begin{table*}[t]
\caption{Characteristics of graphs used in the experiments}
\begin{center}
\begin{tabular}{llll}

\hline

 Size& Graph Name &     $|V|$ &    $|E|$   \\

\hline
Small&  ca-sandi\_auths &   86 &    124    \\

    & inf-USAir97 &   332 &    2126    \\

    &   ca-netscience & 379 &   914  \\

    &  bio-celegans &  453 &    2025 \\
\hline

Medium&      bio-diseasome &   516 &    1188 \\

     & soc-wiki-Vote &   889 &      2914 \\

     & bio-yeast &   1458 &    1948 \\

     & inf-power &   4941 &    6594    \\
\hline
  Large  &  bio-dmela &   7393 &    25569 \\

     & ca-HepPh &   11204 &    117619 \\
      
    &  ca-AstroPh &   17903 &    196972 \\

    &  soc-brightkite &   56739 &  212945 \\

\hline
\end{tabular}
\end{center}
\end{table*}

\section{Experimental Results}

To evaluate the performance of the proposed algorithms, we have compared the algorithms with the state-of-the-art algorithms. In the experiments, we have carried out the comparisons on the efficiencies of the algorithms for enumerating connected induced subgraphs of small cardinality $k$ and large cardinality $k$ respectively. For a fair comparison, all the algorithms including $RSSP$\cite{3}, $Simple$\cite{1}, $Simple-Forward$\cite{7}, $Pivot$\cite{8}, and the proposed algorithms were implemented in the same language-Java. The source codes are available at GitHub\footnote{The source code of the algorithms is available at https://github.com/miningsubgraphs/algorithms}. All experiments were performed on an Intel Core i7 2.90GHz CPU with 8 Gbytes main memory. The reported running times include the time of writing the enumerated sets to the hard drive.

A set of graphs from the Network Repository is used in our limited study\footnote{The dataset used in the experiments is available at http://networkrepository.com}\cite{10}. Similar to \cite{7}, we group the benchmark graphs into three subsets: small graphs with $n<500$, medium-size subgraphs with $500\leq n<5000$, and large graphs with $n\geq 5000$. The characteristics of the benchmarks are shown in Table 2. The columns $|V|$ and $|E|$ indicate the number of vertices and the number of edges of each graph respectively. In the experiments, we set a running time threshold of $600$ seconds. The symbol $-$ in the following tables indicates that the running time exceeds the given threshold of 10 minutes. In the experiments, we noticed that the running time is not steady when the overall running time is very small. Thus, we record the average running time of three runs for the instances with a running time less than one second.

\begin{table*}[t]
\caption{Running time of the algorithms implemented in Python and Java on graph $bio-celegans$ for enumerating all connected induced subgraphs of cardinality $k$  (in seconds)}
\begin{center}
\resizebox{\textwidth}{!}{ 
\begin{tabular}{llllllllll}

\hline

Algorithms &    k=2 &    k=3 &   k=4 & k=5  & k=6 & k=450 & k=451 & k=452\\

\hline
  $Simple-Python (Iteration)$ &   0.007 &       0.263 & 14.926 & - & - & - &135.736 &0.884\\

  $Simple-Forward-Python (Iteration)$ &   0.006 &  0.240 & 13.987 & - & -& - &58.433 &0.338\\

       $Simple-Java (Recursion)$ & 0.011 &   0.087 & 1.775 & 83.843 & - &356.458 &2.746&0.064\\

      $Simple-Forward-Java (Recursion)$ &  0.011 &      0.084 & 1.643 & 77.371 &  - & 294.588 &2.682&0.060\\

\hline
      $|CIS(G,k)|$ &   2025 &     72605 & 3806083 &195573511 & - &14194614&97014&441\\

\hline
\end{tabular}
}
\end{center}
\end{table*}

\begin{table*}[!htb]\footnotesize
\renewcommand\arraystretch{0.7}

\caption{Running time (in seconds) of $Simple$,  $Simple-Forward$, $VSimple$ and $TopDown$  for enumerating all connected induced subgraphs of small cardinality $k$ ($2\leq k\leq 6$)}

\begin{center}

\begin{tabular}{lllllll}
\hline

 Benchmark &$|CIS(G,k)|$ &Cardinality  &$Simple$ & $Simple-Forward$  &  $VSimple$  &  $TopDown$   \\
\hline

                 &   124   &   $k$=  2 &  0.003   &  0.003   &  \textbf{0.003}   &   0.016 \\
                 &   379   &   $k$=  3 &  0.006  &   0.006  &   \textbf{0.006}  &   0.168 \\
 ca-sandi\_auths  &   1422  &   $k$=  4 &  0.008   &  0.008   &  \textbf{0.008}   &   1.407 \\
                 &   5740   &   $k$=  5 & 0.017    & 0.016    &\textbf{0.016}    &   15.986 \\
                 &   23718   &   $k$=  6 & 0.048    &  0.046   & \textbf{0.044}    &   183.949 \\
\hline

                 &   2126   &   $k$=  2 &   0.012  &  0.011   &  \textbf{0.010}   &  0.601  \\
                 &   67827   &   $k$=  3 &  0.106   &  0.103   &  \textbf{0.088}   &  31.785  \\
 inf-USAir97     &   2269621   &   $k$=  4 & 1.176    &  1.116   &  \textbf{1.089}   &  -  \\
                 &   68484518   &   $k$=  5 &  32.454   & 31.790    &  \textbf{29.325}    &  -  \\
                 &   -   &   $k$=  6 &  -   &   -  &   -  &  -  \\
\hline

                 &   914  &   $k$=  2 &   0.007  &   0.006  &   \textbf{0.005}  &  0.370  \\
                 &   4575   &   $k$=  3 &  0.012   &  0.012   &  \textbf{0.012}   &  26.584  \\
 ca-netscience   &   31665   &   $k$=  4 & 0.053    &  0.050   &  \textbf{0.045}   &  -  \\
                 &   244418   &   $k$=  5 & 0.176    &  0.169   & \textbf{0.164}    &  -  \\
                 &   1917058   &   $k$=  6 &  0.878   &  \textbf{0.830}   &  0.861   &  -  \\
\hline

                 &   2025   &   $k$=  2 &  0.011   &  0.011   &  \textbf{0.009}   &   0.935 \\
                 &   72605   &   $k$=  3 &  0.087   &  0.084   & \textbf{0.074}    &  80.201  \\
 bio-celegans    &   3806083   &   $k$=  4 & 1.775    &  1.643   & \textbf{1.517}    &  -  \\
                 &   195573511   &   $k$=  5 & 83.843    & 77.371    &  \textbf{72.221}   &  -  \\
                 &   -   &   $k$=  6 &  -   &  -   &   -  &  -  \\
\hline

                 &   1188   &   $k$=  2 &  0.007   &   0.007  &  \textbf{0.007}   & 0.826   \\
                 &   6758   &   $k$=  3 &  0.017   &  0.016   &   \textbf{0.014}  &  90.191  \\
 bio-diseasome   &   65695   &   $k$=  4 &  0.069   &  0.066   &  \textbf{0.063}   &  -  \\
                 &   765557   &   $k$=  5 &  0.408   &  0.394   &  \textbf{0.376}   &   - \\
                 &   9062333   &   $k$=  6 &  4.121   &  3.895   &   \textbf{3.758}  &  -  \\
\hline

                 &  2914    &   $k$=  2 &  0.011   &  0.012   &  \textbf{0.010}   &  6.250  \\
                 &  45680    &   $k$=  3 &  0.071   &  0.066   &  \textbf{0.059}   &   - \\
 soc-wiki-Vote   &  1121962    &   $k$=  4 &  0.639   & 0.626    &   \textbf{0.623}  &   - \\
                 &  31308165    &   $k$=  5 & 16.994    & 16.061    &  \textbf{14.758}   &  -  \\
                 &  892820902    &   $k$=  6 & 499.905    & 495.006    &  \textbf{430.704}   &  -  \\
\hline

                 &  1948    &   $k$=  2 &  0.009   &  0.009   &  \textbf{0.008}   &  17.728  \\
                 &  11524    &   $k$=  3 &  0.029   & 0.027    & \textbf{0.023}   &  -  \\
 bio-yeast       &  105733    &   $k$=  4 &  0.096   &  0.091   & \textbf{0.082}    & -   \\
                 &  1104980    &   $k$=  5 & 0.536    & 0.509    & \textbf{0.504}    & -   \\
                 &  11718959    &   $k$=  6 &  4.985   & 4.914    & \textbf{4.724}    & -   \\
\hline

                 &  6594    &   $k$=  2 &  0.018   &  0.017   &  \textbf{0.016}   &  -  \\
                 &  17631    &   $k$=  3 & 0.046    &  0.045   &  \textbf{0.042}   &  -  \\
 inf-power       &  63401    &   $k$=  4 & 0.076    &  0.070   &   \textbf{0.068}  &  -  \\
                 &  268694    &   $k$=  5 &  0.207   &  0.206   &  \textbf{0.181}  & -   \\
                 &  1260958    &   $k$=  6 &  0.705   &  0.701   &  \textbf{0.680}  & -   \\
\hline

                 &  25569    &   $k$=  2 &  0.058   &  0.056   &  \textbf{0.052}   & -   \\
                 &  575169    &   $k$=  3 &  0.462   &  0.431   &  \textbf{0.407}   & -   \\
 bio-dmela       &  20943036    &   $k$=  4 & 16.080    &  13.708   &  \textbf{11.699}   & -   \\
                 &  -    &   $k$=  5 &  -   &  -   &  -   &  -  \\
                 &  -    &   $k$=  6 &  -   &  -   &  -   & -   \\
\hline

                 &  117619    &   $k$=  2 &   0.252  &  0.242   &  \textbf{0.239}   &  -  \\
                 &  8560145    &   $k$=  3 &  13.944   & 13.781    &   \textbf{13.352}  &  -  \\
 ca-HepPh        &  -    &   $k$=  4 &  -   &  -   &  -   & -   \\
                 &  -    &   $k$=  5 &  -   &  -   &  -   & -   \\
                 &  -    &   $k$=  6 &  -   &  -   &  -   & -   \\
\hline

                 &  196972   &   $k$=  2 &  0.375   &  0.336   & \textbf{0.306}  & -   \\
                 &  10044854    &   $k$=  3 &  14.485   &   13.915  &  \textbf{10.952}   & -   \\
 ca-AstroPh      &  -    &   $k$=  4 &  -   & -    &  -   &  -  \\
                 &  -    &   $k$=  5 &  -   & -    &  -   & -   \\
                 &  -    &   $k$=  6 &  -   & -    &  -   & -   \\
\hline

                 &   212945   &   $k$=  2 &  0.328   &  0.342   & \textbf{ 0.268 }  &  -  \\
                 &   12432832   &   $k$=  3 & 16.935    &  11.132   &  \textbf{10.963}   & -   \\
 soc-brightkite  &   -   &   $k$=  4 &  -   &  -   &  -   &  -  \\
                 &   -   &   $k$=  5 &  -   &  -   &  -   & -   \\
                 &   -   &   $k$=  6 &  -   &  -   &  -   &  -  \\
\hline
\end{tabular}
\end{center}
\end{table*}

First of all, we performed a straightforward comparison. We have downloaded the source code of $Simple$ and $Simple-Forward$ provided by the authors of \cite{7}. These codes were implemented in Python. We also implemented $Simple$ and $Simple-Forward$ in Java. Please note that, in the Java version, the algorithms are implemented recursively (in the Python version, $Simple$ and $Simple-Forward$ were implemented iteratively), and we adopted the data structures proposed in this paper for $Simple$ and $Simple-Forward$. In the experiments, the Python codes were compiled and executed in PyCharm 2021.2.2 (Community Edition) with Python 3.10.0. The Java codes were compiled and executed in Eclipse 2020-09 with JDK-15.0.1. Preliminary experimental results showed that the algorithms implemented in Java with the proposed data structures can achieve orders of magnitude speedup over the algorithms implemented in Python. As an example, Table 3 shows the running time of the algorithms implemented in Java and Python on the benchmark graph  $bio-celegans$.The number of all connected induced subgraphs of cardinality $k$ in the graph is given in the row $|CIS(G,k)|$. Since the algorithms implemented in Java with the proposed data structures are more efficient, we present the comparison results of algorithms implemented in Java in the following section.

\subsection{Enumerating All Connected Induced Subgraphs of Small Cardinality $k$}
In the implementation, we adopted three basic containers provided in Java (ArrayList, BitSet, and HashSet) for storing the elements respectively. The experimental results reveal that implementation with ArrayList which uses a dynamic array for storing the elements is overall the most efficient one. Our preliminary experiments and the experimental results presented in \cite{7} showed that $Simple$ and $Simple-Forward$ outperform $RSSP$ and $Pivot$ for enumerating connected induced subgraphs of cardinality $k$ on almost every instance. Thus, in this section, we only report the execution time of the proposed algorithms, $Simple$ and $Simple-Forward$ for enumerating all connected induced subgraphs of small cardinality $k \in\{2,3,4,5,6\}$.

As expected, the four algorithms produced the same set of all connected induced subgraphs of small cardinality $k$. The running time of the algorithms for enumerating all connected induced subgraphs of small cardinality $k$ from each graph presented in Table 1 is shown in Table 4. The running time is measured in seconds. A large amount of enumerated subgraphs also reveals that the enumeration of all connected induced subgraphs is indeed a computationally difficult task.

From Table 4, we can see that $VSimple$ is overall the fastest algorithm among all the four algorithms. The speedup achieved by $VSimple$ over $Simple$ or $Simple$-$Forward$ is mainly due to fewer computations carried out in the proposed algorithm. To ensure that each subgraph is enumerated exactly once, $Simple$ and $Simple-Forward$ expand a connected induced subgraph $P$ by adding a neighbor vertex with a larger label than that of originated vertex and it must be not in the closed neighborhood of $C$. The update of the set of the closed neighborhood of $C$ ($N[C]=\cup_{v\in C}N(v)$) should be carried out in each iteration of the recursive calls. The restoration of this set is also required before entering the next iteration. This is more time-consuming than maintaining a guard set $Y$ in $VSimple$. $VSimple$ uses a set $Y$ to record the considered vertices. Only a simple insertion operation at the end of $Y$ is required (line 18, Algorithm 1), and the restoration of $Y$ is only performed before returning to the parent node of the enumeration tree (when all the iterations in the current recursive call are finished). $TopDown$ is less efficient than the other three algorithms on the instances of enumerating induced subgraphs with small cardinality $k$. The main reason is that $TopDown$ incurs $n-k$ recursive calls (the depth of the enumeration tree), however, the bottom-up algorithms only incur $k$ recursive calls. Furthermore, the modification of the adjacent lists of the graphs after deleting a vertex is required in $TopDown$. This modification is relatively inefficient.

\begin{table*}[!htb]\footnotesize

\caption{Running time (in seconds) of $Simple$,  $Simple$-$Forward$, $VSimple$ and $TopDown$  for enumerating all connected induced subgraphs of large cardinality $k$ }

\begin{center}

\begin{tabular}{lllllll}
\hline

 Benchmark &$|CIS(G,k)|$ &Cardinality  &$Simple$ & $Simple-Forward$  &  $VSimple$  &  $TopDown$   \\
\hline

                 &   36407   &   $k$=  83 &  0.303   &  0.290   &  0.252   &  \textbf{0.126} \\
ca-sandi\_auths   &   1837   &   $k$=  84 &  0.0410  &   0.038  &   0.039  &  \textbf{0.020} \\
                 &   61  &   $k$=  85 &  0.0006   &  0.006   &  0.006   & \textbf{0.004} \\

\hline

                 &  4685705    &   $k$=  329 &  94.120   &  73.245   &  83.589   &  \textbf{39.795}  \\
   inf-USAir97   &  46371    &   $k$=  330 &   1.044  &  0.959   &  1.039   &  \textbf{0.401}  \\
                 &  305    &   $k$=  331 &  0.047   &  0.038   &  0.042   &   \textbf{0.020} \\
               
\hline

                 &  5512665    &   $k$=  376 &  113.597   &  108.127   &   101.72   & \textbf{49.899}   \\
 ca-netscience   &  51681    &   $k$=  377 & 1.244    &  1.191  &   1.161  & \textbf{0.499}   \\
                 &  322    &   $k$=  378 &   0.044  &  0.044   &  0.041   & \textbf{0.025}   \\
\hline

                 &  14194614   &   $k$=  450 &  356.458   &  294.588   & 319.415    &  \textbf{181.777}  \\
 bio-celegans    &  97014    &   $k$=  451 &  2.746   &  2.682   & 2.554    &  \textbf{1.004}  \\
                 &  441    &   $k$=  452 &   0.064  &  0.060   &  0.062   &  \textbf{0.030}  \\
\hline

                 &  10914883    &   $k$=  513 &  312.366   &  296.099   &   292.690  & \textbf{150.889}   \\
 bio-diseasome   &  81422    &   $k$=  514 &  2.661   &  2.632   &  2.420   & \textbf{0.990}   \\
                 &  404    &   $k$=  515 &   0.062  &  0.060   &  0.060   &  \textbf{0.028}  \\
\hline

                 &  -    &   $k$=  886 &   -  &  -   &  -   & -   \\
soc-wiki-Vote    &  263965    &   $k$=  887 &   16.446  &  14.525   &  16.294   &  \textbf{5.428}  \\
                 &  727    &   $k$=  888 & 0.121    &  0.113   & 0.120    &  \textbf{0.052}  \\
\hline

                 &  -    &   $k$=  1455 &   -  &  -   &  -   & -   \\
 bio-yeast       &  558202    &   $k$=  1456 &   58.376  &  54.322   &  53.295   & \textbf{33.123}  \\
                 &  1057    &   $k$=  1457 & 0.195    &  0.195   & 0.184    &  \textbf{0.080}  \\
\hline

                 &  -    &   $k$=  4938 &  -   &   -  & -    &  -  \\
 inf-power       &  -    &   $k$=  4939 &  -   &  -   & -    &  -  \\
                 &  3712    &   $k$= 4940 & 1.782    &  1.819   &  1.659   &  \textbf{0.602}  \\
\hline

                 &  -    &   $k$=  7390 &   -  &   -  &   -  & -   \\
 bio-dmela       &  -    &   $k$=  7391 &   -  &  -   & -    & -   \\
                 &  6184    &   $k$=  7392 &  4.871   &  4.308   &  5.247   &  \textbf{1.348}  \\
\hline

                 &  -   &   $k$=  11201 &  -   &   -  &  -   &  -  \\
 ca-HepPh        &  -    &   $k$= 11202 & -    &  - &  -   & -   \\
                 &  10082    &   $k$= 11203 & 16.46    &  11.774   &   18.187  &  \textbf{4.652}  \\
\hline

                 &  -    &   $k$=  17900 &   -  &   -  &  -   & -   \\
 ca-AstroPh      &  -    &   $k$=  17901 &   -  &   -  &  -   &  -  \\
                 &  16836    &   $k$=  17902 &  47.925   &  38.774   &  57.252   &  \textbf{23.192}  \\
\hline

                 &  -    &   $k$=  56736 &  -   &   -  &   -  &  -  \\
 soc-brightkite  &  -    &   $k$=  56737 & -    &   -  &  -   &  -  \\
                 &  44033    &   $k$=  56738 & 415.790    &  288.238   &  432.875   &  \textbf{102.730}  \\
\hline
\end{tabular}
\end{center}
\end{table*}

\subsection{Enumerating All Connected Induced Subgraphs of Large Cardinality $k$}

In this experimental comparison, $k$ is set to $|V|-3$, $|V|-2$ and $|V|-1$ respectively. The running time of $VSimple$, $Simple$, $Simple-Forward$, and $TopDown$ on enumerating all connected induced subgraphs of large cardinality $k$ is shown in Table 5. We can observe that $TopDown$ outperforms the other three algorithms. On average, $TopDown$ is roughly 2.3 times, 2.0 times, and 2.4 times faster than $VSimple$, $Simple$, and $Simple$-$Forward$ for solved instances respectively. It is worthy and interesting to mention that $VSimple$ is faster than $Simple-Forward$ for some instances, and $Simple-Forward$ is more efficient than $VSimple$ for the other instances. After careful analysis of the characteristics of tested instances, we found that $VSimple$ is more efficient when the instance is a sparse graph (e.g. $\frac{|E|}{|V|}<3$, such instances includes ca-sandi-auths,ca-netscience,bio-diseasome,bio-yeast, and inf-power).

In the experiments, we also noticed that all the algorithms can only output the results within $10$ minutes for the large instances with $k=|V|-1$. To figure out which part of the enumeration is time-consuming, we disabled the instruction of writing the sets to the hard disk. In other words, we only record the running time of searching for the solutions without outputting the solutions. Figure 5 (a) and (b) compare the searching time and the outputting time (the time for outputting the solutions to hard disk) of $TopDown$ and $Simple-Forward$ respectively on large instances when $k = |V|-1$. From the comparison, we can see that the search time of $TopDown$ is only a small part of the overall running time. However, for $Simple-Forward$, the searching time is more than half of the overall running time. This also indicates that writing the solutions to the hard disk is the most time-consuming task in $TopDown$ on large instances when $k$ is large. Comparing the searching time, we can see that $TopDown$ is at least 5x and up to 19x faster than $Simple-Forward$ on the four benchmarks when $k=|V|-1$.

\begin{figure*}[t]
  \centering
  \includegraphics[width=0.65\hsize=0.8]{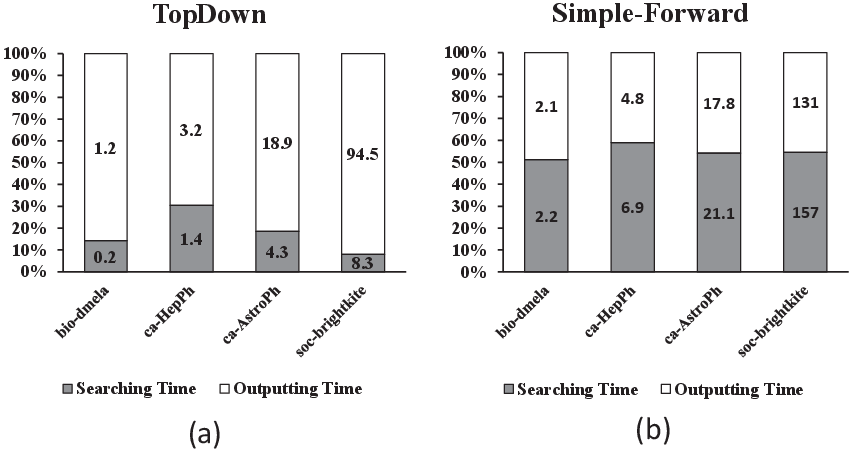}
  \caption{Comparing the searching time and outputting time.(a) time consumption of $TopDown$ (b) time consumption of $Simple$-$Forward$.}
  \label{f2}
\end{figure*}

Since outputting sets of large cardinality is very time-consuming, we propose to output the complement of the enumerated subgraphs when $k>n/2$. Figure 6 presents the speedup achieved by outputting the complement sets over outputting the original sets on the four large benchmarks ($k=|V|-1$). According to the results from Figure 6, it can be seen that outputting the complement sets is more efficient. The achieved speedup is ranging from 3x to 12x. Another significant benefit of outputting the complement sets instead of the original sets when the cardinality is large is that we can save huge amounts of storage space. As an example, for enumerating the connected subgraphs of cardinality $|V|-1$ from the network $soc-brightkite$, storing the complement sets requires only 510 KB, while storing the original sets takes up 15.8 GB. Thus, we also compared the running time of the algorithms in the case of outputting the complement sets. Figure 7 presents the speedup achieved by $TopDown$ over $Simple$-$Forward$ when outputting the complement sets on the four large benchmarks. The achieved speedup is ranging from 5x to 19x.

\begin{figure*}[t]
  \centering
  \includegraphics[width=0.38\hsize=0.8]{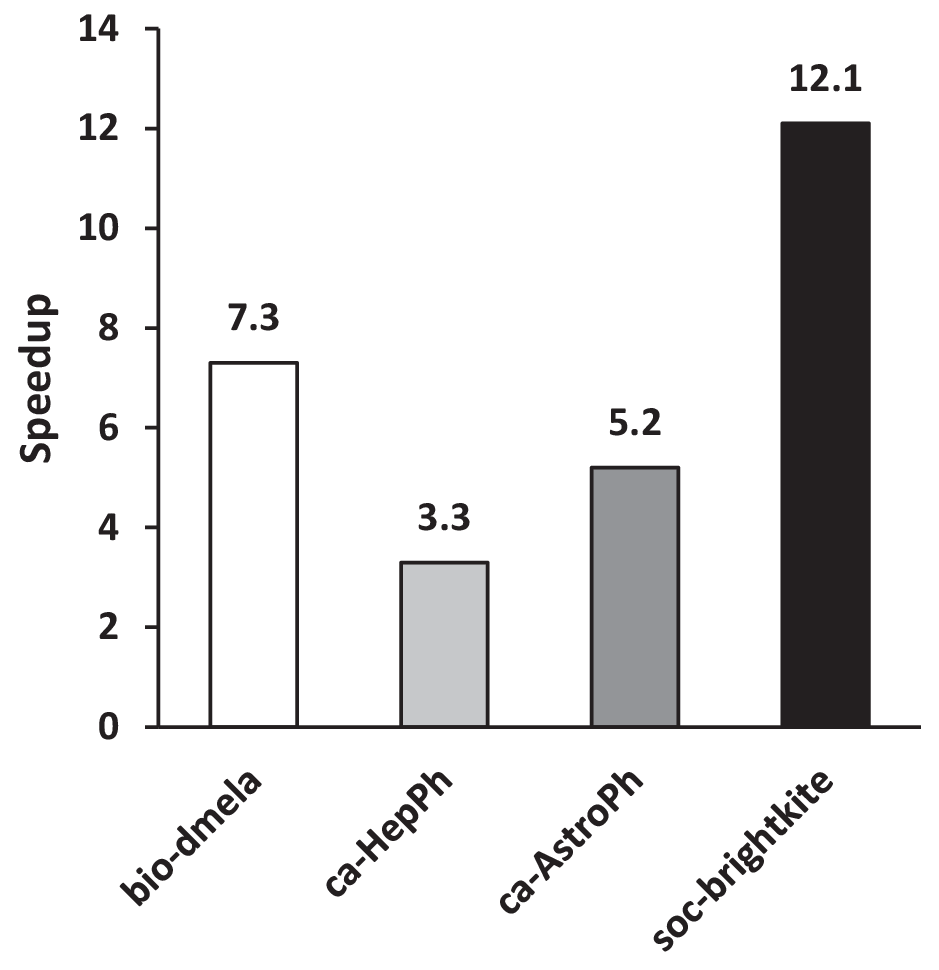}
  \caption{The speedup achieved by outputting the complement sets over outputting the original sets on the four large benchmarks using $TopDown$ ($k=|V|-1$).}
  \label{f2}
\end{figure*}

\begin{figure*}[t]
  \centering
  \includegraphics[width=0.38\hsize=0.8]{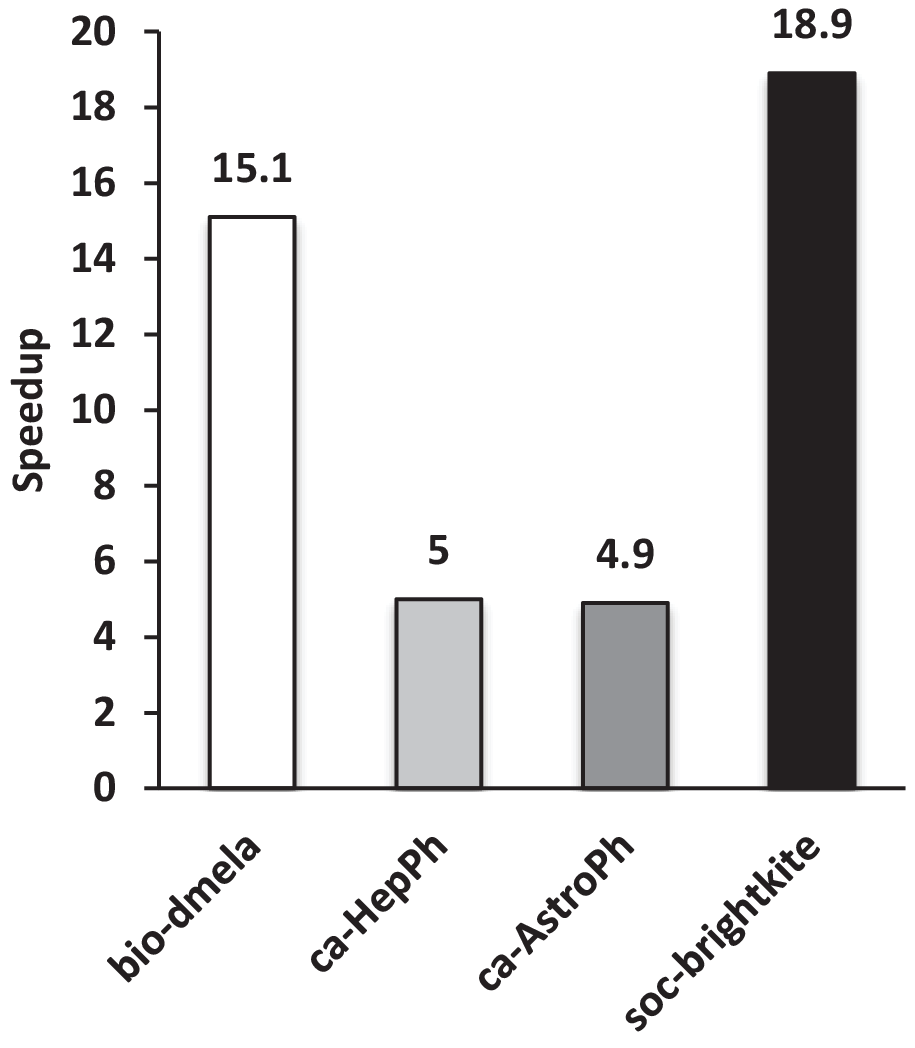}
  \caption{The speedup achieved by $TopDown$ over $Simple$-$Forward$ when outputting the complement sets on the four large benchmarks ($k=|V|-1$).}
  \label{f2}
\end{figure*}


Based on the experimental results of enumerating all connected induced subgraphs of small cardinality $k$ and enumerating all connected induced subgraphs of large  cardinality $k$, we can conclude that $VSimple$ outperforms the other algorithms for enumerating connected induced subgraphs of small cardinality, and $TopDown$ is the fastest algorithm for enumerating connected induced subgraphs of large cardinality.







\section{Discussions}

In this paper, we have presented a bottom-up algorithm and a top-down algorithm for enumerating connected induced subgraphs with cardinality $k$. The bottom-up algorithm is beneficial for instances with small $k$, and the top-down is more efficient for instances with large $k$. A sandwich approach that combines the advantage of the two approaches may be of interest for $k$ is close to $n/2$ in future work. A great deal of work has been done recently on graph mining, scalable graph mining using distributed computing is one of the most promising approaches \cite{13,14}. As mining all the connected induced subgraphs of a given cardinality from a large graph requires large amounts of computing and memory resources, applying parallel computing to speed up the mining process is necessary for large graphs. Similar to the work in \cite{13}, we can partition the graph into several subgraphs of similar size, the connected induced subgraphs of a given cardinality of the different partitions can be computed independently using the enumeration algorithms. The parallelization of mining all the connected induced subgraphs of a given cardinality can be part of our future work.

In \cite{7}, the authors proved the delay of $Simple$ by introducing a new pruning rule. The pruning rule can avoid unnecessary recursions, and it is essential to prove the delay of $Simple$. In this paper, we only presented the running time of $TopDown$. Our work will continue to find such a pruning rule and finally provide the delay of $TopDown$. Uehara has given an upper bound of $n\cdot\frac{(e\Delta)^{k}}{(\Delta-1)k}$ on the number of connected induced subgraphs of cardinality $k$ \cite{6}. Bollob\'{a}s has also presented another upper bound of $(e(\Delta-1))^{(k-1)}$ on the number of connected induced subgraphs of cardinality $k$ containing some vertex $v$\cite{15}. With this bound, Komusiewicz further showed that the overall number of connected induced subgraphs of $k$ is $(e(\Delta-1))^{(k-1)}\cdot(n/k)$ \cite{7,16}. The aforementioned bounds on the number of connected induced subgraphs of cardinality $k$ are all monotonically increasing with $k$. However, we can observe that the number of connected induced subgraphs of large cardinality $k$ (e.g., $k$ is close to $n$) may be decreased with the increase of $k$. Thus, it would be interesting to provide a bound that is better in line with actual trends.

\section{Conclusion}
In this paper, two algorithms that work in a bottom-up manner and a top-down manner respectively are proposed for mining connected induced subgraphs of a given cardinality. The bottom-up algorithm extends the subgraphs by absorbing neighbor vertices until the size of the subgraph is $k$. The top-down algorithm enumerates all the connected induced subgraphs by deleting so-called deletable vertices from large subgraphs. Experiments on real-world networks with different sizes confirmed the efficiency of the proposed algorithms. The proposed top-down algorithm can achieve up to 19x speedup over the state-of-the-art algorithms in cases of large cardinality. From an engineering point of view, the proposed algorithms are very practical as they can be used as the backbone of mining the subgraphs with interesting properties from networks in the area of bioinformatics, networks, and information retrieval. Future work will focus on theoretical perspectives by providing the delay of the top-down algorithm and the upper bound on the number of connected induced subgraphs of cardinality $k$.

\label{last-page}
\end{multicols}
\label{last-page}

\begin{thebibliography}{1}
\bibitem{6}
Uehara R. The number of connected components in graphs and its applications. Manuscript.
\end{thebibliography}
\end{document}